\documentclass[11pt]{amsart}
\usepackage{amsmath}
\usepackage[dvips]{epsfig}
\usepackage{slashed}
\usepackage{accents}
\usepackage{color}
\usepackage{leftidx}
\usepackage{verbatim}

\theoremstyle{plain}
\newtheorem{theorem}{Theorem}

\newtheorem{lemma}[theorem]{Lemma}
\newtheorem{proposition}[theorem]{Proposition}

\newtheorem{prop}[theorem]{Proposition}

\numberwithin{equation}{section}

\DeclareMathOperator{\ric}{Ric}

\begin{document}

\title[Nonexistence of Degenerate Horizons in Static Vacua] {Nonexistence of Degenerate Horizons in Static Vacua and Black Hole Uniqueness}

\author[Khuri]{Marcus Khuri}
\address{Department of Mathematics\\
Stony Brook University\\
Stony Brook, NY 11794, USA}
\email{khuri@math.sunysb.edu}

\author[Woolgar]{Eric Woolgar}
\address{Department of Mathematical and Statistical Sciences\\
University of Alberta\\
Edmonton, AB, Canada T6G 2G1}
\email{ewoolgar@ualberta.ca}


\begin{abstract}
We show that in any spacetime dimension $D\ge 4$, degenerate components of the event horizon do not exist in static vacuum configurations with positive cosmological constant. We also show that without a cosmological constant asymptotically flat solutions cannot possess a degenerate horizon component. Several independent proofs are presented. One proof follows easily from differential geometry in the near-horizon limit, while others use Bakry-\'Emery-Ricci bounds for static Einstein manifolds.
\end{abstract}
\maketitle

\section{Introduction}
\label{sec1}
\setcounter{equation}{0}
\setcounter{section}{1}

\noindent In the classical proof of static vacuum black hole uniqueness, the last case to be considered was that in which degenerate components of the event horizon were present. As was shown in \cite{Chrusciel,ChruscielReallTod}, such black hole configurations
cannot occur. This result applies to the 4-dimensional setting with vanishing cosmological constant $\Lambda=0$. The authors in \cite{ChruscielReallTod} also obtained certain restrictions in the higher dimensional setting and in the presence of a nonzero cosmological constant, but were ultimately unable to extend their result to these situations. The purpose of the present paper is to do just that for $\Lambda\geq 0$. The main result is as follows.

\begin{theorem}\label{theorem1}
\leavevmode
\begin{itemize}
\item [(i)] There do not exist static vacuum black holes having a degenerate horizon component in the presence of a positive cosmological constant $\Lambda> 0$.
%
\item [(ii)] A complete solution of the static vacuum equations with $\Lambda=0$ can have no more than one connected component of a degenerate horizon.
%
\item [(iii)] A solution of the static vacuum equations with $\Lambda=0$ and having an asymptotically flat end\footnote{This refers to the standard notion of asymptotic flatness within the black hole uniqueness context, see eg. \cite{GibbonsIdaShiromizu1}.}
%
%
cannot have a degenerate horizon component.
\end{itemize}

\end{theorem}


An immediate consequence of part (iii) is a generalized version of the classical static black hole uniqueness result. In dimensions $D>4$ the uniqueness proofs \cite{GibbonsIdaShiromizu,GibbonsIdaShiromizu1,Hwang} rely on the positive mass theorem and require all horizon components to be nondegenerate. Here we have shown that degenerate components do not exist, which when combined with \cite{GibbonsIdaShiromizu,GibbonsIdaShiromizu1,Hwang} leads to the following statement.

\begin{theorem}\label{theorem2}
In any dimension $D\geq 4$, an asymptotically flat static vacuum black hole
is isometric to a Schwarzschild-Tangherlini solution.
\end{theorem}


Let $(M^n,g)$ be a Riemannian manifold of dimension $n\geq 3$ on which a function $\varphi$ is defined. Consider the associated static spacetime $(\mathbb{R}\times M^n,G)$ where the spacetime metric takes the form
\begin{equation}
\label{eq1.1}
G=-e^{-2\varphi}dt^2 + g.
\end{equation}
It is assumed that the lapse function is positive on $M^n$ (hence we write it as $e^{\varphi}$), and vanishes on the topological boundary $\partial M^n=\overline{M}^n\setminus M^n$ which itself should be a compact smooth manifold. The vacuum Einstein equations
\begin{equation}
\label{eq1.2}
\ric(G)-\frac12 R(G) G +\Lambda G = 0 ,
\end{equation}
are equivalent to the following set of equations on $M^n$
\begin{equation}
\begin{split}
\label{eq1.3}
&\ric(g) +\mathrm{Hess}_g \varphi -d\varphi\otimes d\varphi = \frac{2\Lambda}{n-1} g , \\
&\Delta_g \varphi -\vert d\varphi \vert^2= \frac{2\Lambda}{n-1} .
\end{split}
\end{equation}
Note that the first equation almost implies the second. Indeed, the first equation together with the twice contracted second Bianchi identity shows that
\begin{equation}
\label{eq1.4}
\Delta_g \varphi -\vert d\varphi \vert_g^2= \frac{2\Lambda}{n-1} +Ce^{2\varphi}\ ,
\end{equation}
where $C$ is a constant. The second equation of \eqref{eq1.3} is recovered in the case that $C=0$.

Recall that a Killing horizon is a null hypersurface defined by the vanishing in norm of a Killing field $V$, which is normal to the horizon. In the static case above $V=\partial_{t}$ and the Killing horizon corresponds to $\mathbb{R}\times\partial M^n$.
Killing horizons come naturally equipped with a notion of surface gravity $\kappa$, defined through the equation
\begin{equation}
\label{eq1.5}
d|V|^2=-2\kappa V.
\end{equation}
A component of the horizon is referred to as degenerate (or extreme) if its surface gravity vanishes $\kappa=0$.

An important observation is that the static vacuum equations can be expressed in terms of the $N$-Bakry-\'Emery-Ricci tensor
\begin{equation}
\label{eq1.6}
\ric_{\varphi}^{N}(g):= \ric(g) +\mathrm{Hess}_g \varphi -\frac{1}{(N-n)}d\varphi\otimes d\varphi\ .
\end{equation}
In general $N$ may take values in the compactified real line, where the last term in \eqref{eq1.6} is removed when $N=\pm\infty$. This expression arises naturally when an $N$-dimensional metric splits as a warped product over $(M^n,g)$. Namely, the $N$-dimensional Ricci tensor splits and its projection onto the base yields $\ric_{\varphi}^N$. The term \emph{synthetic dimension} for $N$ arises since, in this context, $N$ is the dimension of the total space.
The first static vacuum equation in \eqref{eq1.3} may be rewritten as
\begin{equation}
\label{eq1.7}
\ric_{\varphi}^{n+1}(g)= \frac{2\Lambda}{(n-1)}g.
\end{equation}
Metrics which satisfy this relation are referred to as \textit{quasi-Einstein} \cite{CaseShuWei}. It turns out that many of the basic Ricci curvature results of Riemannian geometry are known to hold as well for the Bakry-\'Emery-Ricci curvature.
In particular we will exploit Bakry-\'Emery versions of Myers' Theorem, the Splitting Theorem, and arguments used in the proof of Synge's Theorem \cite{Petersen}.
It is the purpose of this paper to introduce these techniques into the study of static black hole uniqueness, and thereby establish Theorem \ref{theorem1}. More precisely, Myers' Theorem and the Splitting Theorem will yield special cases of Theorem \ref{theorem1} in Section \ref{sec3}, and in Section \ref{sec4} the Synge type methods will produce a full proof. Section \ref{sec2} is dedicated to recording technical results concerning degenerate horizons for use in later sections.
We note that the theory associated with Bakry-\'Emery-Ricci curvature has previously been applied to study static solutions of the Einstein equations which are geodesically complete, in \cite{Case,Reiris,Reiris1}.

\subsection*{Acknowledgements} M Khuri acknowledges the support of NSF Grant DMS-1708798. E Woolgar was supported by a Discovery Grant RGPIN 203614 from the Natural Sciences and Engineering Research Council. Both authors thank the Erwin Schr\"odinger International Institute for Mathematics and Physics and the organizers of its ``Geometry and Relativity'' program, where this paper was conceived. We thank Piotr Chru\'sciel for a comment on an earlier draft.

\section{Degenerate Components of the Horizon}
\label{sec2} \setcounter{equation}{0}
\setcounter{section}{2}

\subsection{Degenerate horizons as asymptotic ends}
Consider a degenerate component of the Killing horizon in a static black hole spacetime. A key prerequisite for application of the Riemannian geometric techniques mentioned in the introduction, is the fact that within the $t=0$ slice such degenerate components lie infinitely far away from any other point. This has been shown in \cite{Chrusciel}, although here we offer a simple proof using Gaussian null coordinates \cite{KunduriLucietti}. These coordinates may be introduced near a degenerate horizon, and give the following form of the spacetime metric
\begin{equation}
\label{eq2.1}
G = 2 dv \left(dr +\frac{1}{2}r^2 F(r,x) dv +rh_a(r,x) dx^a\right)
 + h_{ab}(r,x) dx^a dx^b.
\end{equation}
Here $V=\partial_{v}$ represents the timelike Killing field, $r=0$ coincides with the horizon, and $h_{ab}$ denotes the induced metric on $\mathcal{H}$ the horizon cross-section. The orbit space metric on a constant time slice $M^n$ is then given by
\begin{equation}
\label{eq2.2}
g_{ij}=G_{ij}-\frac{G_{iv}G_{jv}}{G_{vv}}=G_{ij}+\frac{G_{iv}G_{jv}}{r^2 |F|}.
\end{equation}
Note that since the Killing field is timelike away from the horizon
\begin{equation}\label{F}
F(r,x)<0\quad\quad\text{ for }\quad\quad r>0.
\end{equation}
It follows that
\begin{align}
\label{eq2.4}
\begin{split}
g(\partial_{r},\partial_{r})&=\frac{G_{rv}^2}{r^2 |F|}=\frac{1}{r^2 |F|},\\
g(\partial_r,\partial_{x^a})&=\frac{G_{rv}G_{av}}{r^2 |F|}=-\frac{rh_a}{r^2 |F|},\\
g(\partial_{x^a},\partial_{x^b})&=G_{ab}+\frac{G_{av}G_{bv}}{r^2 |F|}=h_{ab}+\frac{h_a h_b}{|F|}.
\end{split}
\end{align}
Let $\gamma(r)=(r,x^{a}(r))$, $r\in[0,r_0]$ be a smooth curve in the orbit space intersecting $\mathcal{H}$, with tangent vector $\dot{\gamma}$. On this curve $|h_a \dot{x}^a|+|F|\leq c$ independent of $r$, as these two functions are continuous on a compact interval. We then have
\begin{equation}\label{estimate}
|\dot{\gamma}(r)|^2=\frac{(1-rh_a \dot{x}^a)^2}{r^2 |F|}+h_{ab}\dot{x}^a \dot{x}^b\geq\frac{1}{2r^2 |F|}
\end{equation}
for $r_0$ sufficiently small, and hence the length of this curve diverges
\begin{equation}
\label{eq2.6}
s(r) = \int_{r}^{r_0}|\dot{\gamma}(r)|dr\geq \frac{1}{2c}\int_{r}^{r_0}\frac{dr}{r}\rightarrow\infty\quad\text{ as }\quad r\rightarrow 0\ .
\end{equation}

\begin{lemma}\label{deghorizon}
A degenerate component of the horizon cross-section is infinitely distant from any point in a constant time slice of a smooth static spacetime.
\end{lemma}

The derivation above used only \eqref{eq2.1} and the timelike nature of the static Killing field in the interior. The result did not require the Einstein equations to hold.

\subsection{Near-horizon geometry}
Typically the geometry of $(M^n,g)$ is asymptotically cylindrical in a neighborhood of a degenerate component of the horizon cross-section, and thus one expects appropriate decay of certain geometric quantities upon approach to $\mathcal{H}$. In order to establish the desired estimates take the near-horizon limit $v\rightarrow \frac{v}{\varepsilon}$, $r\rightarrow \varepsilon r$, and $\varepsilon\rightarrow 0$, which produces the near-horizon geometry metric
\begin{equation}
\label{eq2.7}
G_{\rm NH} = 2 dv \left(dr +\frac{1}{2}r^2 F(x) dv +rh_a(x) dx^a\right)
 + h_{ab}(x) dx^a dx^b.
\end{equation}
If $G$ satisfies the vacuum Einstein equations then the near-horizon data $(F, h_a,h_{ab})$
solve the near-horizon geometry equations on $\mathcal{H}$
\begin{equation}
\label{eq2.8}
\begin{split}
R_{ab}=&\, \frac{1}{2}h_a h_b -\nabla_{(a}h_{b)}+\Lambda h_{ab},\\
F=&\, \frac{1}{2}|h|^2-\frac{1}{2}\nabla_{a}h^a+\Lambda,
\end{split}
\end{equation}
where $R_{ab}$ denotes the Ricci tensor associate with metric $h_{ab}$. If $\Lambda\geq 0$, then integrating the second equation in \eqref{eq2.8} and using the divergence theorem yields
\begin{equation}\label{intF}
\int_{\mathcal{H}}F=\int_{\mathcal{H}}\left(\frac{1}{2}|h|^2
+\Lambda\right)\geq 0,
\end{equation}
since it is assumed that $\mathcal{H}$ is compact without boundary. In light of \eqref{F} it must be the case that $F(x)\leq 0$, and hence \eqref{intF} shows that $F(x)\equiv 0$.
In fact we immediately obtain the following nonexistence result when $\Lambda>0$ and strong restrictions when $\Lambda=0$.

\begin{prop}\label{proposition4}
In any spacetime dimension $D\geq 4$, there do not exist static vacuum black holes with $\Lambda>0$ and having a degenerate horizon component. Moreover, if $(F, h_a,h_{ab})$ is the near-horizon data of a degenerate horizon component in a static vacuum black hole
with $\Lambda=0$, then $F=h_a=0$ and $h_{ab}$ is Ricci flat.
\end{prop}

The derivation given above leading to this result made no use of Lemma \ref{deghorizon}.
Moreover, it should be pointed out that although the Schwarzschild-de Sitter black hole
may have a degenerate horizon, when this occurs the hypotheses of Proposition \ref{proposition4} are not satisfied. More precisely, in Schwarzschild-de Sitter $e^{-2\varphi}=1-2m/r-(\Lambda/3)r^2$ and the various horizons occur at the zeros of this function. In order for there to be three real roots the mass must lie in the interval $m\in\left(-\frac{1}{3\sqrt{\Lambda}},\frac{1}{3\sqrt{\Lambda}}\right)$. Degenerate horizons only occur when $m$ agrees with one of the endpoints of the interval, or equivalently two of the roots coincide. However, in both of these cases the static Killing field fails to be timelike near the horizon.

The second part of this proposition concerning the case $\Lambda=0$ has been independently proved in \cite{ChruscielReallTod}. It implies the nonexistence of such static vacuum solutions with a degenerate horizon component when $D=4$, as the Ricci flat condition is not compatible with the topological restrictions \cite{KunduriLucietti} on extreme horizons in this case.

\subsection{The vanishing of $F$}
While we have seen that the $\Lambda=0$ near-horizon equations lead us to conclude that $F$ vanishes on approach to degenerate horizons in vacuum spacetimes, the set-up of subsection 2.1 leads to a more general result for all $\Lambda\ge 0$.

To see this, consider a smooth curve $\gamma(s)=(r(s),x^{a}(s))$, parameterized by arclength, which as above connects a point interior to $M^n$ to the degenerate horizon cross-section $\mathcal{H}$. Recall that $-2\varphi=\log(-|\partial_v|^2)=\log(r^2|F|)$, and let $\dot{\gamma}=d\gamma/dr$. Then
\begin{align}
\label{eq2.10}
\begin{split}
\partial_{s}(\varphi\circ\gamma)=&(\partial_{r}\varphi) \frac{dr}{ds}+(\partial_{x^a}\varphi)\frac{dx^a}{ds}\\
=&\, -\frac{\partial_{r}\varphi}{|\dot{\gamma}|}-\frac{\dot{x}^a\partial_{x^a}\varphi}
{|\dot{\gamma}|}\\
=&\, \frac{\partial_{r}\log\left(r\sqrt{|F|}\right)}{|\dot{\gamma}|}
+\frac{\dot{x}^a \partial_{x^a}\log\left(r\sqrt{|F|}\right)}{|\dot{\gamma}|}\\
=&\, \frac{\partial_{r}\left(r\sqrt{|F|}\right)}{r\sqrt{|F|}|\dot{\gamma}|}
+\frac{r\dot{x}^a \partial_{x^a}\sqrt{|F|}}{r\sqrt{|F|}|\dot{\gamma}|}.
\end{split}
\end{align}
According to \eqref{estimate} we have $r\sqrt{|F|}|\dot{\gamma}|\geq 1/\sqrt{2}$ for small $r$, and as above $|\dot{x}^a|\leq c$. It follows that
\begin{equation}
\label{eq2.11}
|\partial_{s}(\varphi\circ\gamma)|=O(\sqrt{|F|}+r)\quad\quad\text{ as }\quad\quad r\rightarrow 0.
\end{equation}
This implies the following result.

\begin{lemma}\label{decay}
Let $(\mathbb{R}\times M^n,-e^{-2\varphi}dt^2+g)$ be a static vacuum spacetime with nonnegative cosmological constant $\Lambda\geq 0$. If $\gamma(s)$ is a smooth curve in the orbit space parameterized by arclength and connecting a point interior to $M^n$ to a degenerate horizon cross-section, then $F\circ\gamma\rightarrow 0$ and $\partial_{s}(\varphi\circ\gamma)\rightarrow 0$ as $s\rightarrow\infty$.
\end{lemma}

\section{Application of Myers' Theorem and the Splitting Theorem}
\label{sec3} \setcounter{equation}{0}
\setcounter{section}{3}

\noindent In the last section, we showed how Proposition \ref{proposition4} (and thus much of Theorem \ref{theorem1}) followed easily from the second of the two equations in \eqref{eq2.8}. In this section, we show that parts (i) and (ii) of Theorem \ref{theorem1} follow easily from the first of the two equations in \eqref{eq2.8}, together with Lemma \ref{deghorizon}, by application of known results from Riemannian geometry.

A classical result in Riemannian geometry asserts that a complete manifold with Ricci curvature bounded below by a uniform positive constant must have finite diameter. This is known as Myers' Theorem \cite{Petersen}. It turns out that such a result is valid when the boundedness condition $\ric \ge c>0$ for the Ricci curvature is replaced by $\ric_{f}^{N}\ge c >0$ for the $N$-Bakry-\'Emery-Ricci tensor, for any twice differentiable $f$ and any $N>n$.

\begin{theorem}[{\cite[Theorem 5]{Qian}}]\label{Qian}
Let $(M^n,g)$ be a complete Riemannian manifold. If for some $N>n$ there exists a $C^2$ function $\varphi$ such that $\ric_{\varphi}^{N}(g)\ge \Lambda g$, $\Lambda>0$, then $(M^n,g)$ has finite diameter.
\end{theorem}

The static vacuum equations imply \eqref{eq1.7}, and therefore the hypotheses of this version of Myers' Theorem are satisfied when the cosmological constant is positive and nondegenerate horizons (i.e., minimal surface boundaries) are not present. It follows that the constant time slices of this solution are of finite diameter. However, if degenerate horizon components were present, this would contradict Lemma \ref{deghorizon}. This yields the next result which partially generalizes the main result of \cite{KhuriWoolgar} in the static case.

\begin{proposition}\label{proposition7}
There do not exist complete static vacuum solutions with $\Lambda>0$ and having a degenerate horizon component.
\end{proposition}

Another basic result for complete Riemannian manifolds of nonnegative Ricci curvature is the Splitting Theorem of Cheeger and Gromoll \cite{Petersen}. This result states that if such a manifold admits a \textit{line}; i.e., a complete geodesic which realizes the distance between any two of its points; then it must isometrically split off a Euclidean factor. Extensions of this theorem have been established in the Bakry-\'Emery setting. For $N=\infty$ this was accomplished by Lichn\'erowicz in \cite[\S 26, pg 90]{Lichnerowicz}. The following statement treats the case of finite $N>n$.

\begin{theorem}[{\cite[Theorem 1.3]{FLZ}}]\label{splitting}
Let $(M^n,g)$ be a complete connected Riemannian manifold with a smooth function $\varphi$ and $N>n$ such that $\ric^N_{\varphi}(g)\ge 0$. If $(M^n,g)$ admits a \emph{line} then it splits isometrically as ${\mathbb E}^l\times \mathcal{N}$ with $l\ge 1$, where $\mathcal{N}$ contains no line and ${\mathbb E}^l$ is Euclidean $l$-space. Furthermore $\varphi$ is constant on ${\mathbb E}^l$, and $\mathcal{N}$ has nonnegative $(N-l)$-Bakry-\'Emery-Ricci curvature.
\end{theorem}

This may be applied to static vacuum black holes as follows.

\begin{proposition}\label{proposition9}
A complete solution of the static vacuum equations with $\Lambda\ge 0$ can have no more than one connected component of a degenerate horizon.
\end{proposition}

\begin{proof}
Suppose that the time slice $(M^n,g)$ has at least two degenerate horizon components. Then Lemma \ref{deghorizon} implies that it is disconnected at infinity, and hence must contain a line \cite[Lemma 41]{Petersen} connecting two horizon components.

As a solution of the static vacuum equations with $\Lambda\geq 0$, $(M^n,g)$ has nonnegative $(n+1)$-Bakry-\'Emery-Ricci curvature. The splitting theorem in \cite{FLZ} now applies to show that $M^n={\mathbb E}^l\times \mathcal{N}$ isometrically, for some $\mathcal{N}$ and $l\geq 1$. Moreover $-e^{-2\varphi}=|\partial_{t}|^2$ is constant on the Euclidean factor $\mathbb{E}^l$ which contains the line. However this is impossible since the Killing field $\partial_{t}$ is timelike on the interior of $M^n$ but null on the horizon.
\end{proof}

\section{Synge Type Arguments and the Proof of Theorem \ref{theorem1}}
\label{sec4} \setcounter{equation}{0}
\setcounter{section}{4}

\noindent In this section we will make use of second variation arguments for geodesics, reminiscent of those used in the typical proof of Synge's Theorem \cite[p 172, Theorem 26]{Petersen} from Riemannian geometry, in order to give an alternative proof of part (i) and establish part (iii) of Theorem \ref{theorem1}. Note that parts (i) and (ii) have already been proved in previous sections.

Consider a minimizing geodesic $\gamma(s)$ in $M^n$ parameterized by arclength $s\in[0,\infty)$, connecting an interior point to a degenerate component of the horizon cross-section. Let $\gamma_{\tau}(s)$, $\tau\in(-\varepsilon,\varepsilon)$ be a 1-parameter family of curves such that $\gamma_0=\gamma$, and denote the variation vector field along $\gamma$ by $X=\partial_{\tau}\gamma_{0}$.
The energy of each curve on a finite interval is defined by
\begin{equation}\label{eq4.1}
E(\tau)=\frac{1}{2}\int_{0}^{s_0}|\partial_{s}\gamma_{\tau}|^2 ds,
\end{equation}
and the second variation formula states that
\begin{equation}\label{eq4.2}
E''(0)=\int_{0}^{s_0}\left(|\nabla_{s}X|^2-\langle R(\partial_{s}\gamma,X)X,\partial_{s}\gamma\rangle\right)ds
+\left.\langle \nabla_{X}X,\partial_{s}\gamma\rangle\right\vert_{s=0}^{s_0}.
\end{equation}
Let $\{e_{i}\}_{i=1}^{n-1}$ denote an orthonormal basis for the orthogonal complement of $\partial_{s}\gamma(0)$ in $T_{\gamma(0)}M^n$, and parallel transport this basis along the geodesic to obtain variation fields $X_i=f(s)e_{i}(s)$ where $f$ is a smooth function on $[0,s_0]$. If $E_i$ denotes the energy associated with this variation, then utilizing the static vacuum equation \eqref{eq1.3} and integrating by parts produces
\begin{align}\label{variation}
\begin{split}
\sum_{i=1}^{n-1}E_{i}''(0)=&\int_{0}^{s_0}\left[(n-1)f'^2
-\mathrm{Ric}(\partial_s \gamma,\partial_s \gamma)f^2\right]ds+\left.f^2\sum_{i=1}^{n-1}\langle
\nabla_{e_i}e_i,\partial_{s}\gamma\rangle\right\vert_{s=0}^{s_0}\\
=&\int_{0}^{s_0}\left[(n-1)f'^2
-\left(\frac{2\Lambda}{n-1}-\nabla_{s}\partial_{s}\varphi+\left ( \partial_{s}\varphi\right )^2\right)f^2
\right]ds+\left.f^2\sum_{i=1}^{n-1}\langle
\nabla_{e_i}e_i,\partial_{s}\gamma\rangle\right\vert_{s=0}^{s_0}\\
=&\int_{0}^{s_0}\left[(n-1)f'^2
-\left(\frac{2\Lambda}{n-1}+\left ( \partial_{s}\varphi\right )^2\right)f^2-2ff'\partial_{s}\varphi
\right]ds\\
&+\left.f^2\left(\partial_{s}\varphi+\sum_{i=1}^{n-1}\langle
\nabla_{e_i}e_i,\partial_{s}\gamma\rangle\right)\right\vert_{s=0}^{s_0}\\
\leq &\int_{0}^{s_0}\left[(n+1)f'^2
-\left(\frac{2\Lambda}{n-1}+\frac12 \left ( \partial_{s}\varphi\right )^2\right)f^2\right]ds
+\left.f^2\left(\partial_{s}\varphi+\sum_{i=1}^{n-1}\langle
\nabla_{e_i}e_i,\partial_{s}\gamma\rangle\right)\right\vert_{s=0}^{s_0}.
\end{split}
\end{align}

The sum of second variations of energy is nonnegative if the variation vector fields vanish at the endpoints; i.e., if $f(0)=f(s_0)=0$. Thus we are motivated to minimize the integral on the right-hand side of \eqref{variation}. Consider the Rayleigh quotient
\begin{equation}\label{eq4.4}
\lambda_1=\inf_{f(0)=f(s_0)=0}\frac{\int_{0}^{s_0}\left(f'^2
-\frac{2\Lambda}{n^2-1}f^2\right)ds}{\int_{0}^{s_0}f^2 ds},
\end{equation}
which gives the principal Dirichlet eigenvalue for the operator $\frac{d^2}{ds^2}+\frac{2\Lambda}{n^2-1}$ on the interval $[0,s_0]$.
A computation shows that this value is $\lambda_1=\pi^2/s_0^2 -2\Lambda/(n^2-1)$. Since $\Lambda>0$, for a sufficiently long interval along the geodesic $\lambda_1<0$. Thus by choosing $f(s)=\sin(\pi s/s_0)$ to be the principal eigenfunction a contradiction is achieved from \eqref{variation}. This proves $(i)$ of Theorem \ref{theorem1}.

For part (iii) the setting is an asymptotically flat static vacuum solution. Assume that it has a degenerate component of the event horizon. Let $\gamma:[0,s_0]\rightarrow M^n$ be a geodesic which minimizes the distance between an $r$-level set (intersected with $M^n$) $\mathcal{H}_r\ni\gamma(0)$ in Gaussian null coordinates near the degenerate component, and a coordinate sphere $\mathcal{S}_{\mathbf{r}}\ni\gamma(s_0)$ in the asymptotically flat end. This geodesic must remain within the interior of $M^n$. To see this, observe that it cannot intersect a nondegenerate horizon component tangentially since such surfaces are totally geodesic, and it cannot intersect
these boundaries transversely since it would not be minimizing. Moreover, for similar reasons $\gamma$ must meet $\mathcal{H}_r$ and $\mathcal{S}_{\mathbf{r}}$ orthogonally. We may now follow the second variation arguments above, choosing a variation $\gamma_{\tau}^{i}(s)$ for each orthogonal variational vector field $X_i=f(s) e_i(s)$ such that $\gamma_{\tau}^{i}(0)\subset\mathcal{H}_r$ and $\gamma_{\tau}^{i}(s_0)\subset\mathcal{S}_{\mathbf{r}}$. Then $E_{i}^{''}(0)\geq 0$. Furthermore, setting $f(s)=e^{-\alpha\varphi\circ\gamma(s)}$ where $0<\alpha<1/\sqrt{2(n+1)}$ yields
\begin{equation}\label{aaa}
(n+1)f'^2
-\frac12 |\partial_{s}\varphi|^2 f^2=\left[(n+1)\alpha^2-\frac{1}{2}\right]e^{-2\alpha\varphi}|\partial_{s}\varphi|^2\leq 0.
\end{equation}
By letting $r\rightarrow 0$ and $\mathbf{r}\rightarrow\infty$ a contradiction is obtained
with \eqref{variation} as follows. The left-hand side of \eqref{variation} is nonnegative, while the integral on the right-hand side tends to a negative number in light of \eqref{aaa}, and the boundary terms converge to zero. This last fact is a consequence of the asymptotically flat fall-off which implies that
\begin{equation}\label{eq4.6}
e^{-2\alpha\varphi}=1+O(\mathbf{r}^{-1}),\quad\quad \partial_{s}\varphi=O(\mathbf{r}^{-1}),\quad\quad\sum_{i=1}^{n-1}\langle
\nabla_{e_i}e_i,\partial_{s}\gamma\rangle=O(\mathbf{r}^{-1})\text{ }\text{ }\text{ as }\text{ }\text{ }\mathbf{r}\rightarrow\infty,
\end{equation}
and the asymptotics
\begin{equation}\label{eq4.7}
e^{-2\alpha\varphi}\rightarrow 0,\quad\quad \partial_{s}\varphi\rightarrow 0,\quad\quad\sum_{i=1}^{n-1}\langle
\nabla_{e_i}e_i,\partial_{s}\gamma\rangle\rightarrow 0\text{ }\text{ }\text{ as }\text{ }\text{ }r\rightarrow 0,
\end{equation}
which result from the vanishing of $|\partial_t|$ at the horizon and Lemma \ref{decay}.
In particular, the last of these limits shows that the mean curvature of $\mathcal{H}_r$ tends to zero. This can be seen from the fact that the horizon cross-section $\mathcal{H}_0$ is a future apparent horizon, and since the $t=0$ slice is time symmetric the mean curvature agrees with the future null expansion for any $r$ so that $H=\theta_{+}\rightarrow 0$.
This completes the proof of Theorem \ref{theorem1}.(iii).

\end{document}